\documentclass[pra,aps,twocolumn,groupedaddress,superscriptaddress,preprintnumbers]{revtex4-2}

\usepackage{graphicx}
\usepackage[nointegrals]{wasysym} %
\usepackage[export]{adjustbox}
\usepackage{amsmath,amsfonts,amssymb,latexsym}
\usepackage{hhline}
\usepackage{bm}
\usepackage{verbatim}
\usepackage{enumitem}
\usepackage{mathrsfs}
\usepackage{slashed}
\usepackage{empheq}

\newcommand{\p}{\partial}

\newcommand{\lan}{\langle}
\newcommand{\ran}{\rangle}

\newcommand{\ra}{\rightarrow}
\newcommand{\lra}{\leftrightarrow}

\newcommand{\uva}{{\mathbf{\hat a}}}

\newcommand{\uvx}{{\mathbf{\hat x}}}
\newcommand{\uvy}{{\mathbf{\hat y}}}

\newcommand{\bfzero}{{\mathbf{0}}}

\renewcommand{\(}{\left(}
\renewcommand{\)}{\right)}
\renewcommand{\[}{\left[}

\newcommand{\mt}{\mapsto}

\newcommand\bpm            {\begin{pmatrix}}
	\newcommand\epm           {\end{pmatrix}}

\newcommand{\bs}{\bigskip}

\def\app#1#2{%
	\mathrel{%
		\setbox0=\hbox{$#1\sim$}%
		\setbox2=\hbox{%
			\rlap{\hbox{$#1\propto$}}%
			\lower1.1\ht0\box0%
		}%
		\raise0.25\ht2\box2%
	}%
}

\newcommand{\tw}{\textwidth}

\newcommand{\ct}{\Theta}

\newcommand{\ope}\odot

\usepackage{manfnt}

\newcommand{\bi}{\begin{itemize}}
	\newcommand{\ei}{\end{itemize}}

\usepackage{amsthm}
\newtheorem{theorem}{Theorem}
\newtheorem{definition}{Definition}

\newtheorem{proposition}{Proposition}

\newtheorem{lemma}{Lemma}

\theoremstyle{definition}

\newcommand\bpro		  {\begin{proposition}}
	\newcommand\epro 		  {\end{proposition}}
\newcommand\bproof			  {\begin{proof}}
	\newcommand\eproof 		  {\end{proof}}

\newcommand\ed            {\end{definition}}

\newcommand\be            {\begin{equation}}
\newcommand\ee            {\end{equation}}
\newcommand\ba            {\begin{aligned}}
\newcommand\ea            {\end{aligned}}
\newcommand\bea{\begin{equation}\begin{aligned}}
	\newcommand\eea{\end{aligned}\end{equation}}
\usepackage{xcolor}

\usepackage{hyperref} 
\hypersetup{final}
\hypersetup{colorlinks, citecolor=red, linkcolor=red, urlcolor=red}

\renewcommand{\ss}{\subsection}

\renewcommand{\a}{\alpha}

\newcommand{\g}{\gamma}

\renewcommand{\S}{\Sigma} 

\renewcommand{\l}{\lambda}

\newcommand{\bfdel}{{\boldsymbol{\delta}}}

\newcommand{\bfa}{\mathbf{a}}

\newcommand{\bfd}{\mathbf{d}}

\newcommand{\bfr}{\mathbf{r}}

\newcommand{\zt}{\mathbb{Z}_2}
\newcommand{\zn}{\mathbb{Z}_N}

\newcommand{\rr}{\mathbb{R}}
\newcommand{\qq}{\qquad}

\newcommand{\zz}{\mathbb{Z}}

\newcommand{\mcb}{\mathcal{B}}

\newcommand{\mcd}{\mathcal{D}}

\newcommand{\mca}{\mathcal{A}}

\newcommand{\mcv}{\mathcal{V}}

\newcommand{\mcr}{\mathcal{R}}

\newcommand{\sfA}{\mathsf{A}}
\newcommand{\sfB}{\mathsf{B}}
\newcommand{\sfC}{\mathsf{C}}
\newcommand{\sfD}{\mathsf{D}}

\newcommand{\sfN}{\mathsf{N}}

\newcommand{\sfP}{\mathsf{P}}

\usepackage[mathscr]{eucal} %

\usepackage{dcolumn}

\usepackage{pifont}
\usepackage[normalem]{ulem}

\usepackage[caption=false]{subfig}
\usepackage{enumitem}
\usepackage{amsthm}
\usepackage{physics}
\usepackage{booktabs}
\usepackage{algorithm}
\usepackage{algpseudocode}
\usepackage{bm}
 
\renewcommand\qq{\qquad}

\newcommand{\maj}{{\sf maj}}
\newcommand{\dam}{{\sf dam}}
\usepackage{seqsplit}
\usepackage{array}
\newcommand{\plog}{{p_{\sf log}}} 

\newcommand{\algorule}{\par\noindent\rule{\linewidth}{0.4pt}\par} 

\begin{document}

	\title{Exploring the Landscape of Non-Equilibrium Memories with Neural Cellular Automata}
	
	\author{Ehsan Pajouheshgar}
	\affiliation{
    IC and SB Schools, École Polytechnique Fédérale de Lausanne (EPFL), Switzerland
    }

    \author{Aditya Bhardwaj}
    \affiliation{Department of Computing and Mathematical Sciences, Caltech, Pasadena, CA 91125, USA}
    \affiliation{Institute for Quantum Information and Matter, Caltech, Pasadena, CA 91125, USA}

    \author{Nathaniel Selub}
    \affiliation{Department of Physics, University of California Berkeley, Berkeley, CA 94720, USA}

    \author{Ethan Lake}
	\email{elake@berkeley.edu}
    \affiliation{Department of Physics, University of California Berkeley, Berkeley, CA 94720, USA}	
    
	\begin{abstract}
		We investigate the landscape of many-body memories: families of local non-equilibrium dynamics that retain information about their initial conditions for thermodynamically long time scales, even in the presence of arbitrary perturbations. In two dimensions, the only well-studied memory is Toom's rule. Using a combination of rigorous proofs and machine learning methods, we show that the landscape of 2D memories is in fact quite vast. We discover memories that correct errors in ways qualitatively distinct from Toom's rule, have ordered phases stabilized by fluctuations, and preserve information only in the presence of noise. Taken together, our results show that physical systems can perform robust information storage in many distinct ways, and demonstrate that the physics of many-body memories is richer than previously realized. Interactive visualizations of the dynamics studied in this work are available at \href{https://memorynca.github.io/2D}{https://memorynca.github.io/2D}.
	\end{abstract}
	
	\maketitle
	
	\paragraph*{Introduction: } 
	The design of many-body systems that retain memory of their initial conditions in the presence of perturbations, {\it i.e.}, systems which function as {\it robust many-body memories}, is a fundamental problem at the intersection of non-equilibrium statistical mechanics \cite{ponselet2013phase}, fault-tolerant computation \cite{gacs2001reliable}, and biological information processing \cite{wagner2013robustness}. Memory in locally-interacting systems must be achieved through collective, distributed error correction: errors must be autonomously detected and corrected by the noisy dynamics itself, and the general mechanisms by which this may be accomplished are only partially understood \cite{gacs1978one, gacs2001reliable,cirel2006reliable,toom1980stable,Harrington2004,mordvintsev2020growing,rakovszky2024defining,ray2024protecting,balasubramanian2024local, paletta2025high,lake2025fast}. 
	
	This work investigates non-equilibrium many-body memories in 2D spin systems. The dynamics we study are designed to remember a single bit of information, which they encode in the sign of the system's magnetization.
	Concretely, we consider dynamics acting on spins $s_\bfr \in \{\pm 1\}$ located on the sites of an $L\times L$ lattice with periodic boundary conditions. The dynamics is defined as a probabilistic cellular automaton (CA), specified by a noise strength $p \in [0,1]$, a noise bias $\eta \in [-1,1]$, and a CA rule $\mca$. Letting $\mcr_\bfr$ denote the $3\times 3$ block of sites centered at a site $\bfr$, the dynamics updates $s_\bfr$ according to a particular function $\mca$ of the spins on $\mcr_\bfr$. When noise is included, the dynamics acts by sending $ s_\bfr \mt \mca(s|_{\mcr_\bfr})$ with probability $1-p$, and $s_\bfr \mt \pm 1 $ with probability $p(1\pm \eta)/2$. 
	
	In the math and computer science literature, CA are normally updated {\it synchronously}: they are run in discrete time, and all sites are updated at once. From the physics perspective, an {\it asynchronous} update scheme---where the sites update at times governed by independent Poisson processes---is often more natural, as it corresponds to a continuous-time Markov process. We will consider both types of updates in this work.

	We say that $\mca$ defines a {\it robust memory} if there is an open ball in the $(p,\eta)$ plane within which the dynamics preserves the sign of the magnetization $m = L^{-2} \sum_\bfr s_\bfr$ for a time that diverges as $L\ra\infty$. Having a robust memory is only possible if regions of errors are corrected in a time proportional to their linear size: this is because a bias $\eta$ favoring a minority domain causes it to ballistically expand, and this expansion rapidly destroys the encoded information if the minority domain is not quickly eroded by the dynamics. In this context, the problem of constructing a robust memory is therefore equivalent to finding a CA that quickly identifies and corrects arbitrarily-shaped domains of errors \footnote{An example of dynamics which does not erase minority domains quickly enough is Glauber dynamics for the 2D Ising model at $T< T_c$. A domain of size $R$ is corrected only in time $\sim R^2$, which leads to the memory being destroyed in the presence of biased noise (or essentially equivalently, in the presence of a small magnetic field). }.
	
	The most well-known robust memory in 2D is Toom’s rule, $\mca_{\sf Toom}(s_\bfr) = \maj(s_\bfr, s_{\bfr+\uvx}, s_{\bfr+\uvy})$, where $\maj(S)$ takes the majority vote of the set $S$ \cite{toom1980stable}. The triangular shape of the majority voting region ensures that minority domains (regardless of their shape) are ballistically eroded starting from their top-right corners. 
	To date, all constructions of robust memories with local order parameters in the physics literature descend from this single mechanism \cite{ray2024protecting,machado2023absolutely,mcginley2022absolutely}. It is then natural to ask: is Toom's rule the only way of achieving a robust 2D memory? 
	
	\begin{figure*}
		\includegraphics[width=\tw]{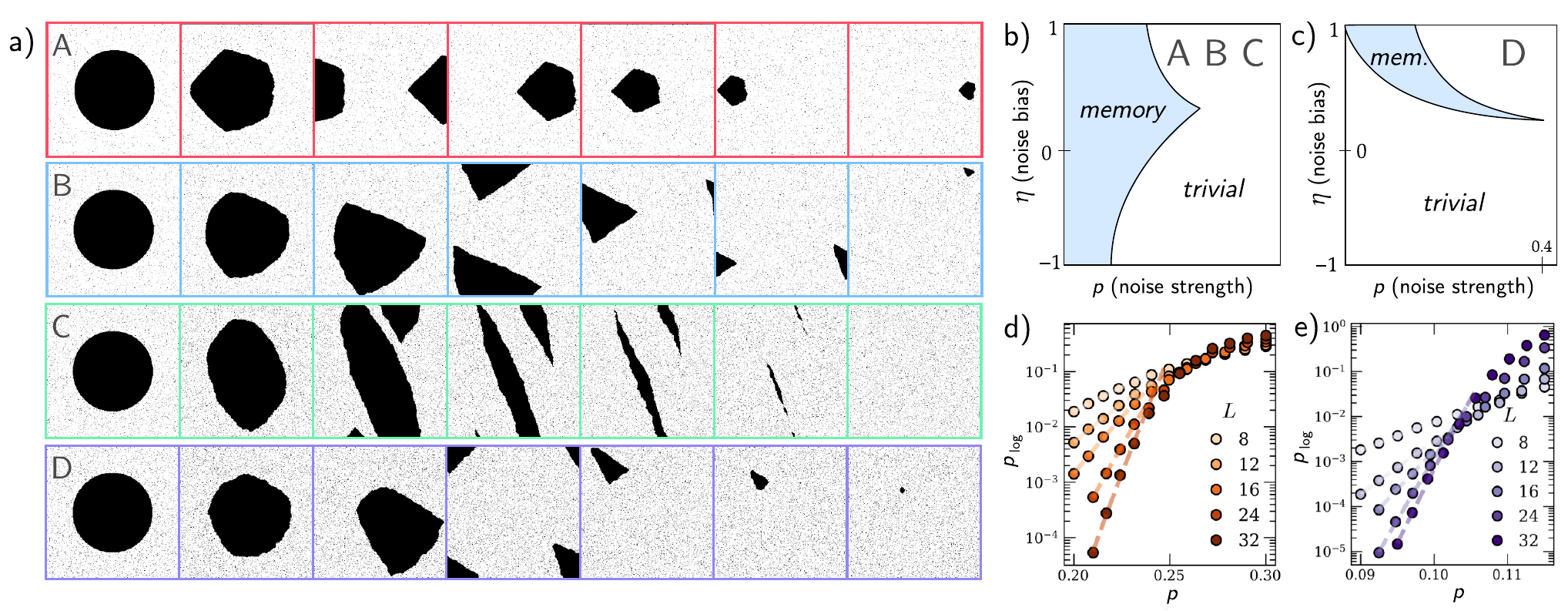} 
		\caption{\label{fig:megafig} Histories showing the correction of large minority domains of $1$ spins (black) for different memories, with the noise bias $\eta=+1$ chosen to maximally favor the error. Time runs left to right. {\sf b)} A schematic phase diagram in the $(p,\eta)$ plane for typical rules found with our NCA architecture. {\sf c)} Schematic phase diagram of the $\sfD$ rule, with a re-entrant noise-stabilized memory phase; the tip of the memory phase is around $(p,\eta) \approx (0.4,0.08)$. ${\sf d)}$ $\plog$ against $p$ for asynchronous updates, shown for the $\sfA$ automaton at $\eta = 0$. The dashed lines are fits to $\plog \propto p^{\l L}$ with $\l\approx1.3$. ${\sf e)}$ the same as ${\sf d}$, but for the ${\sfC}$ automaton at $\eta = -1$; here $\l\approx2.6$. }
	\end{figure*}
	
	\paragraph*{A first class of examples:} as a first exact result, we prove that the triangular geometry of Toom's rule is in fact not special, and that most ways of doing local majority voting yield robust memories under synchronous updates:
	
	\begin{theorem}[asymmetric majority voters are robust] \label{thm:majtheorem} 
		Suppose $\mca(s|_{\mcr_\bfr}) = \maj(\{s_{\bfr' \in \mcv_\bfr}\})$ performs a majority vote on a subset of spins in the neighborhood $\mcv_\bfr \subset \mcr_\bfr$. Then $\mca$ is a robust memory under synchronous updates iff $\mcv_\bfr$ is \textbf{not} invariant under a $\pi$ rotation $C_\pi$ about any lattice point. 
	\end{theorem} 
	The proof makes use of a theorem of Toom \cite{toom1980stable} regarding {\it monotone} cellular automata (viz. those for which $\mca$ defines a monotonic Boolean function \footnote{A function $f(x) \, : \, \{-1,1\}^n \ra \{-1,1\}$ is monotonic if increasing the number of 1s in the input cannot change an output of 1 to an output of $-1$.}, such as a majority vote), which states that any monotone eroder is a robust memory at small enough noise strengths. We prove erosion by proving that minority domains are always eroded ballistically along certain directions if $\mcv_\bfr$ has no $C_\pi$ symmetry. Conversely, with $C_\pi$ symmetry, we show that there are always certain shapes of minority domains that, regardless of their size, are never corrected by $\mca$; under biased noise, these domains grow and rapidly destroy the encoded information. The details are provided in App.~\ref{app:majproof}. 
	
	As a concrete example, consider the CA rule $\sfA$ defined by the voting region $\mcv_\bfr = \bfr + \{\uvy,\uvx+\uvy,\uvx,-\uvy,-\uvx\}$.  As illustrated in the top row of Fig.~\ref{fig:megafig}~a, $\sfA$ erodes minority domains by mapping them to a series of increasingly small pengaton-shaped regions. 
	% An interactive visualization of the error-correcting dynamics of this and the other automata studied in this work is available at \href{https://memorynca.github.io/2D}{https://memorynca.github.io/2D}. 

	As stated, Theorem~\ref{thm:majtheorem} only applies to dynamics with synchronous updates. The extension to asynchronous updates is nontrivial, as there exist systems which retain memory only when updated synchronously \cite{squeezing}. To establish results for asynchronous dynamics, we therefore turn to numerics. We use Monte Carlo simulations to estimate the {\it logical failure rate}, defined as the probability for a magnetization flip to occur by time $O(L)$ when the system is initialized in its least stable memory state: 
	\be p_{\sf log} = \max_{s_\pm} \sfP[{\sf maj}(\mca^{cL}(s)) \neq \maj(s)],  \ee 
	where $s_\pm$ are respectively the all-$1$s and all-$(-1)$s states, and $c$ is a constant, fixed at $c = 5$ in our numerics. For any memory we must have $\lim_{L \ra \infty} p_{\sf log} = 0$, and for the CA under consideration we expect, for small enough $p$, the scaling 
	\be \label{plogscaling} p_{\sf log} \sim (p/p_c)^{\l L}\ee 
	for constants $p_c,\l$ \footnote{This is due to the fact that for all of the memories considered in the main text, flipping $\ct(2L)$ spins along two non-contractible loops that span each cycle of the torus is sufficient to cause a magnetization flip.}.
	The logical failure rate $\plog$ for $\sfA$ dynamics with asynchronous updates and bias $\eta = 0$ is plotted in Fig.~\ref{fig:megafig}~d, where we observe a threshold crossing at $p_c \approx 26\%$---significantly higher than the $p_{c,{\sf Toom}}\sim 14\%$ of Toom's rule \cite{bennett1985role,ray2024protecting}.

	\paragraph*{Neural cellular automata: } Do there exist robust memories beyond asymmetric majority voters? Addressing this question via brute-force search is impossible: the number of CA rules taking $3\times 3$ blocks of spins as input is $2^{2^9}$---vastly greater than the number of atoms in the observable universe---and random rules are empirically observed to have a negligibly small chance of being robust memories. To more efficiently explore the space of possible dynamics, we use a machine-learning-guided approach inspired by the neural cellular automaton (NCA) framework~\cite{mordvintsev2020growing}. Using this setup, we find robust memories {\it beyond} asymmetric majority voters. In this setting, $\mca$ is initially parameterized as a random continuous-state CA rule (operating on $\rr$-valued spins), implemented by a small neural network. The neural network is optimized using gradient descent to minimize a loss function that is designed to measure $\mca$'s ability to erase large domains of errors in the presence of noise. Including large error regions in the training data, instead of relying on their spontaneous occurrence from noisy dynamics, enables faster and more targeted learning. After a fixed number of optimization steps,  we discretize $\mca$ to $\zt$-valued spins. Further details on the NCA architecture and training are provided in App.~\ref{app:nca}.
	%	Directly exposing the model to large domains of errors (in contrast to waiting for such errors to appear randomly under a particular noise model) allows a more efficient way to train the model. 
	
	In a set of 1000 training runs with different random seeds, 37 cellular automata converged to robust memories. None of these 37 memory rules implemented majority votes, none were symmetric under $s_\bfr \mt -s_\bfr$ \footnote{$\mca$ is symmetric under spin-flip if $\mca(-s) = -\mca(s)$ for all input states $s$.}, and all but 6 were monotonic. Symmetric rules occurred only outside the memory class, where they were common: out of the 387 symmetric rules observed, all were monotonic, and none were memories. Monotonicity was similarly widespread beyond the memory class, with 994 monotonic rules in total.
	%	, but relatively rare among memories themselves. 
	In most cases where the learned rule was not a memory, the dynamics was observed to perform something similar to a $C_\pi$-symmetric majority vote.  Finally, while training was performed only with synchronous updates, all of the learned memory rules were observed to remain robust after making the updates asynchronous.

	Representative examples of learned memories are shown in the last three rows of Fig.~\ref{fig:megafig}~a. The automaton $\sfB$ in the second row erodes domains by shrinking them into small triangle shapes, while $\sfC$ in the third row squeezes domains along an oblique angle. 
	Both ${\sfB,\sfC}$ are monotonic, and possess no symmetries. 
	These rules have phase diagrams and a scaling of $\plog$ qualitatively similar to that of $\sfA$, and their thresholds are tabulated in Tab.~\ref{tab:thresholds}.

	\begin{table}[htbp]
		\centering
		\renewcommand{\arraystretch}{1.2}
		\begin{tabular}{c@{\hspace{1em}}c@{\hspace{1em}}c@{\hspace{1em}}c@{\hspace{1em}}c}
			
			$p_c$ (\%)& ${\sf A}$ & ${\sf B}$ & ${\sf C}$ & ${\sf D}$ \\
			\toprule 
			async. & 4/26/4&  4/22/3 & 8/26/10 & noise stab. %1/9/3 
			\\ sync. & 5/32/5 & 6/32/5 & 9/34/11  
			&noise stab. \\ 
			\midrule
			\bottomrule
		\end{tabular}
		\caption{Approximate thresholds (in percent) of the memories studied in the main text, for both asynchronous (top) and synchronous (bottom) updates. The three values $x/y/z$ indicate the maximum noise probability thresholds at biases of $\eta = 1, 0, $ and $-1$, respectively. ``noise stab.'' indicates that the memory is present only at nonzero noise; see the main text. }
		\label{tab:thresholds}
	\end{table}
	
	\paragraph*{Noise-stabilized order: }
	The automaton $\sfD$ shown in the fourth row of Fig.~\ref{fig:megafig}~a, which is one of the non-monotonic rules found in training, is rather unusual. At $p=0$ it has exponentially many absorbing states, created by inserting small islands of $-1$s in a background of $1$s (in the language of \cite{rakovszky2024defining,gacs_eroder_slides}, it is thus {\it not} a linear eroder). 
	Under  unbiased noise these islands percolate, resulting in a unique steady state dominated by $-1$s. With a modest amount of noise biased to favor $1$, however, the noise and CA dynamics render the islands unstable. This results in the phase diagram of Fig.~\ref{fig:megafig}~c, featuring a ``re-entrant'' memory phase in an intermediate range of $p$. The $\sfD$ automaton is thus a {\it noise-stabilized memory}, and its existence demonstrates that even non-eroding automata can define robust memories. 
	This phenomenon can in some sense be regarded as a non-equilibrium analogue of the Pomeranchuk effect \cite{richardson1997pomeranchuk,poulin2019self}, but arises due to a quite different microscopic mechanism. The only other example of noise-stabilized order the authors are aware of is the recent Ref.~\cite{marsan2025perturbed}, which is based on a modification of Gacs' complex 1D automaton \cite{gacs2001reliable}.

	\paragraph*{The failure of mean-field theory: } The most natural first pass at theoretically understanding the phase diagrams observed so far is to perform a mean-field (MF) analysis. This is done by writing down a master equation for $\p_t \lan s_\bfr\ran$ (with the expectation value $\lan \cdot\ran$ taken in the non-equilibrium steady state), and assuming that correlation functions factorize as $\lan s_\bfr s_{\bfr'} \ran \approx \lan s_\bfr \ran \lan s_{\bfr'} \ran \equiv m^2$, where we have assumed $\lan s_\bfr\ran$ is independent of $\bfr$ (as it is for all of the steady states we consider). Doing this yields (see App.~\ref{app:meanfield}) 
	\be \p_t m = p\eta + (1-p) \mca_{MF}(m) - m \equiv - \p F_{\rm eff} / \p m\ee 
	where $\mca_{MF}(m)$ is obtained by writing down $\mca(s|_{\mcr_\bfr})$ as a polynomial in the spins on $\mcr_\bfr$, and substituting $s_\bfr \mt m $. 
	
	MF theory predicts an ordered phase when the effective free energy $F_{\rm eff}$ develops multiple distinct minima. 
	%An expression for $F_{\rm eff}$ in terms of $\mca$'s rule table is derived in App.~\ref{app:meanfield}. 
	For memories where $\mca$ performs a majority vote, $F_{\rm eff}$ is a conventional-looking free energy with multiple minima at $p<p_{c,MF}$, and for  the $\sfA$ automaton we obtain $p_{c,MF} = 7/15 > p_c$. Having $p_{c,MF} > p_{c}$ is of course unsurprising: since MF theory neglects fluctuations, it is expected to always {\it over}estimate the tendency towards ordering. 
	
	Interestingly,  the CA rules ${\sfB,\sfC,\sfD}$ found in training are observed to be {\it unstable} in MF: for these rules, $F_{\rm eff}$ always has a unique minimum, regardless of $p$. These are the first examples known to the authors of systems where MF theory entirely misses an ordered phase, and incorrectly predicts a disordered phase for all $p$. The order in these systems is thus {\it fluctuation-stabilized} \footnote{This phenomenon might be called ``order by disorder,'' but is quite distinct from what is usually meant by this phrase, which conventionally refers to situations where (fine-tuned) systems with many degenerate ground states have their degeneracy lifted by fluctuations. Conventional ``order by disorder'' {\it decreases} the number of degenerate states, and thus in some sense actually reduces the amount of order. By contrast, the dynamics considered in this work uses fluctuations to {\it increase} the number of steady states, thereby increasing the amount of order.  }, and  future work \cite{squeezing} will show how fluctuations beyond MF stabilize the ordered phase. 
	
	{\it Synchronicity transitions: }
	By explicitly modifying their rule tables, the learned rules can be used as jumping-off points for exploring different types of error-correcting dynamics. As one illustrative example, we discuss a  modification that changes how these automata behave with respect to changes in synchronicity. 
	
	\begin{figure}
		\includegraphics[width=.5\tw]{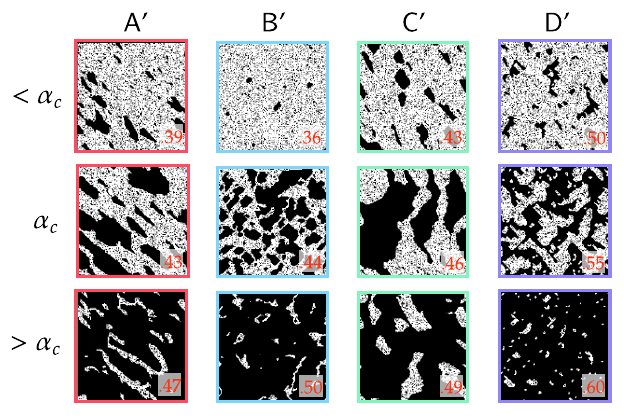} 
		\caption{\label{fig:synch} Synchronicity transitions in the modified CA rules $\sfA',\sfB',\sfC',\sfD'$. Dark colors are $1$s, and each image shows the state of a system after 50 steps of noiseless dynamics following a quench from the all-$(-1)$s state. The number at bottom right of each panel indicates the value of $\a$, with $\a_c$ the critical value of the synchronicity parameter. For rules $\sfA',\sfB',\sfC'$, the all-1s state is rapidly reached when $\a > \a_c$. For $\sfD'$, the patterns shown in the bottom-right panel persist for infinite times in the absence of noise.} 
	\end{figure}
	In the absence of noise, all of the learned NCA rules have the all 1s and all $-1$s configurations as absorbing states.
	Consider destroying the $-1$ absorbing state by modifying the CA rule to send $\mca(s|_{\mcr_\bfr}) \mt +1$ if all spins in $\mcr_\bfr$ are $ -1$.
	For the examples considered above, this change immediately destroys the memory for synchronous updates, even in the absence of noise: an initial state with mostly $-1$s is immediately mapped to a state with mostly $1$s, leading to an irrecoverable loss of the encoded information. On the other hand, when the updates are asynchronous, isolated errors introduced by the modification to the transition rule have a high probability of being corrected by updates on adjacent sites,  and the memory is retained up to a non-zero (but smaller) critical noise strength. 
	These automata thus define {\it asynchronicity-protected} memories, and stand in contrast with all other known memory CA, which become {\it more} robust when the updates are made synchronous. 
	
	To further explore this phenomenon, consider updates where each site has a probability $\a \in [0,1]$ of being updated at each time step. At small enough $p$, each of the studied memories undergoes a percolation transition at an intermediate update fraction $\a_c$ between a robust memory phase (at small $\a$), and a trivial phase dominated by $1$s (at large $\a$); this phenomenon is somewhat similar to the synchronicity transitions that occur in other (non-robust) CA, such as Conway's game of life \cite{bersini1994asynchrony,fates2007asynchronism}. As shown in Fig.~\ref{fig:synch}, the spatial spin patterns produced across this transition can be quite rich. 
	
	\begin{figure} 
		\includegraphics[width=.48\tw]{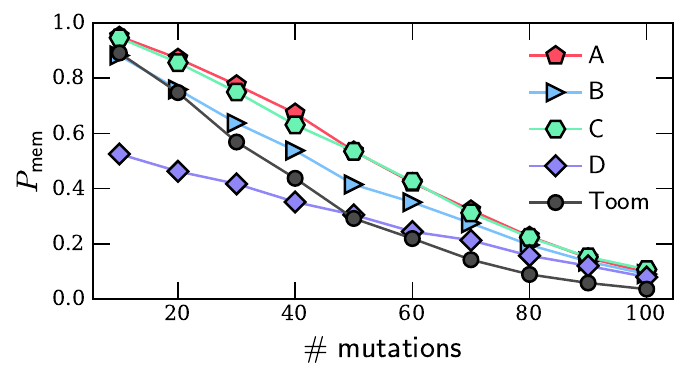} 
		\includegraphics[width=.48\tw]{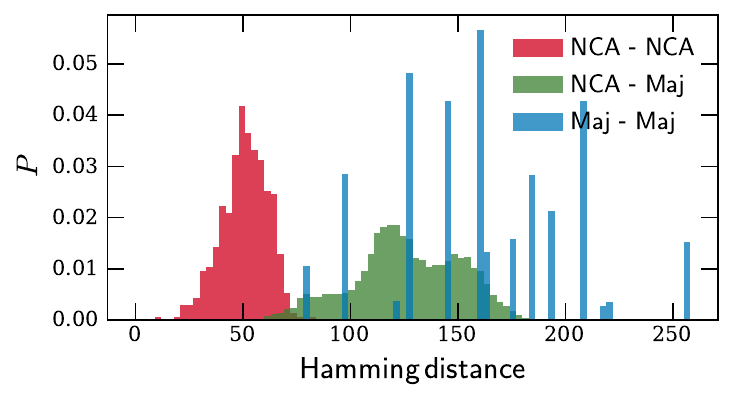} 
		\caption{	\label{fig:mutations} {\it Top:} Probability that a rule remains a memory $P_{\sf mem}$ as a function of the number of ``mutations'' made to its rule table; each data point is based on sampling 5000 randomly mutated rules. {\it Bottom:} Hamming distances between pairs of distinct rule tables for the learned NCA and majority-vote memories. Each point is the distance between one rule from category A and one from category B.} 
	\end{figure} 
	
	\paragraph*{The rule landscape: } 
	The modification just discussed was chosen in a quite special way, and it is natural ask what happens when {\it random} modifications are made. To study this, note that a given CA is determined by a binary string of length $2^9$, with each digit specifying the output of the CA rule on one of the $2^9$ possible inputs. We then take this binary string, which may be regarded as that CA's ``genetic code'', and subject it to ``mutations'' by randomly flipping bits. In the top panel of Fig.~\ref{fig:mutations}, we plot the probability $P_{\sf mem}$ that the studied rules remain a memory as a function of the number of inflicted mutations, where $P_{\sf mem}$ is estimated as the fraction of mutated rules which remain able to shrink both types of minority domains in the absence of noise. Rules ${\sfA,\sfB,\sfC}$ are seen to be quite robust to mutations, with $\gtrsim 50$ mutations required to bring $P_{\sf mem}$ below $0.5$. 
	
	This result shows that each learned CA rule carves out an extended region of robust memories within the space of all CA rules. How are the regions associated to different rules connected with one another? We address this by computing the Hamming distances between the rule tables of each learned NCA rule, as well as the distances between these rules and those defined by asymmetric majority votes. The results in the bottom panel of Fig.~\ref{fig:mutations} show that the NCA rules lie in a bubble of Hamming radius $\sim50$, and appear to be distinctly separated from the rules based on majority votes. Fleshing out the topology of this landscape is left to future work. 
	
	\paragraph*{Outlook:} The phenomena discovered in this work indicate that the physics of non-equilibrium memories is quite diverse, and far from fully understood. Immediate future research directions include generalizing the NCA framework to $\zn$ and $U(1)$ memories, investigating the phase transitions that occur when the memory is lost, and theoretically characterizing the distinct universality classes of local error-correcting dynamics. 

    While our focus in this work has been entirely on classical memories, machine learning methods can also be used to discover local error-correcting dynamics for {\it quantum} memories. The simplest examples are closely related to one-dimensional classical memories, which remain poorly understood, although in this context, the (reasonable) assumption of noiseless classical error-correcting hardware makes the task easier. Future work will describe this in detail \cite{rlpaper}. 
	
	We close with some rather speculative remarks. The systems studied in this work are capable of ``homeostasis'': the robust protection of (here, a single bit of) information against degradation by noise. It is natural to wonder about mechanisms by which richer behavior, like nontrivial information processing and environmental adaptation, can robustly emerge in locally-interacting systems.
	%	, as they are observed to emerge in biological organisms. 
	It is well known that simple CA rules can display remarkably life-like emergent behavior
	%, including morphogenesis, structure maintance, and types of evolutionary behavior
	\cite{langton1997artificial,chan2018lenia,mordvintsev2020growing},
	but in none of these cases has the problem of noise resilience been carefully addressed, and an outstanding problem for theory is to identify the complexity threshold that must be crossed in order for CA to exhibit complex {\it and  robust} emergent behavior. The constructions of Gacs \cite{gacs2001reliable,gacs1989self}---whose design calls to mind the genetic code of an elaborate artificial organism---provide extremely complicated starting points, but there is clearly much more to be done. We expect the NCA approach developed in the present work to be useful in this regard. 
    
	\paragraph*{Acknowledgments: } 
	
	We thank Shankar Balasubramanian, Sarang Gopalakrishnan, Hsin-Yuan Huang, Vedika Khemani, Yaodong Li, Tibor Rakovszky, and Mike Zaletel for discussions and feedback.

        A.B. is supported by the Kortschak Scholars Program.
	E.L. is supported by a Miller Research Fellowship. 
        N.S. is supported by the National Science Foundation Graduate Research Fellowship Program under Grant No. 2146752. Any opinions, findings, and conclusions or recommendations expressed in this material are those of the authors and do not necessarily reflect the views of the National Science Foundation.

    \paragraph*{Author contributions:} 

    A.B., E.L., and N.S. conceived the idea of using NCA to discover robust memories. A.B., E.P., and N.S. designed and implemented the numerics, E.P. built the interactive demo, and E.L. performed the theoretical calculations. 	
	
	\bibliography{refs}
		
	\onecolumngrid
	\begin{center}
    \bigskip 
    \bigskip 
		\textbf{\large Exploring the Landscape of Non-Equilibrium Memories with Neural Cellular Automata: Supplementary information}\\[.2cm]
	\end{center}
	
	\setcounter{equation}{0}
	\setcounter{figure}{0}
	\setcounter{table}{0}
	\setcounter{page}{1}
	\renewcommand{\theequation}{S\arabic{equation}}
	\renewcommand{\thefigure}{S\arabic{figure}}
	\renewcommand{\bibnumfmt}[1]{[S#1]}
	\renewcommand{\citenumfont}[1]{S#1}
	
	\ss*{Contents} 
	\begin{itemize}
		\item Section~\ref{app:otherrules}: provides explicit CA rule tables for the rules studied in the main text. 
		\item Section~\ref{app:nca}: describes in detail the neural cellular automaton architecture and training process. 
		\item Section~\ref{app:majproof}: proves theorem~\ref{thm:majtheorem} from the main text. 
		\item Section~\ref{app:meanfield}: provides details of the mean-field analysis quoted in the main text.  
	\end{itemize}
	
	\bs

	\section{More information about the rule landscape}\label{app:otherrules}
	
	In this short appendix, we provide the explicit rule strings for the CA rules studied in the main text, as well as a few other rules of interest. 
	Each deterministic rule may be parametrized by a binary string of length $2^9 = 512$. Let $w$ denote a particular configuration of a $3\times 3$ block of spins, with coordinates labeled by $\{(x,y) \in \{-1,0,1\}^2\}$ and each $w_{(x,y)} \in \{-1,1\}$. We use the convention where the $n_w$th digit of $\mca$'s rule string is equal to $\frac{1+\mca(w)}2$, where the number $n_w\in\{1,\dots,512\}$ is defined as 
	\be  \label{rulestring} n_w = 1+\sum_{x=-1}^1 \sum_{y=-1}^1 \frac{1+w_{(x,y)}}2 2^{3(x+1) + (y+1)}.\ee 
	The rule strings for the rules studied in the main text, as well as a few other interesting examples, are provided in  table~\ref{tab:rule_landscape}.

	\section{Neural Cellular Automata: Architecture and Training}\label{app:nca}
	
	In this appendix, we describe in more detail the neural cellular automaton (NCA) framework discussed in the main text. 
	
	An NCA operates on lattices of generalized spins, where each generalized spin contains \texttt{c} components, which we will refer to as ``channels.'' The first ``visible'' channel encodes the logical information that the NCA attempts to preserve, while the remaining ``hidden'' channels are used to facilitate the preservation of this information. Each channel is represented by a real number during training, but is projected onto a discrete spin valued in $\pm 1$ after training is completed. A discrete spin value of $+1$ is encoded as the real number $+\Delta$ while $-1$ is encoded as $-\Delta$, with $\Delta=0.25$ in our experiments. 
    % Each spin in the NCA is free to use the full $[-1,1]$ range, or even exceed it, during training. 
    The role of $\Delta$ is to define the target states that the NCA must preserve, and the memory loss is computed as a combination of L1 and L2 distances from the visible channel to either $+\Delta$ or $-\Delta$, depending on the stored target spin. 
	
	\subsection{Architecture}
	We now discuss the NCA architecture in more detail.
	Classical NCAs employ convolutional filters such as the Sobel filter and the Laplacian to extract information from the neighboring cells. This way, the learned dynamics resemble partial differential equations acting on smooth spatial fields \cite{noisenca}. In our setting, however, both the internal state and space 
	are discrete in nature (spins), so we dispense with gradient-like filters and instead directly feed the raw Moore neighborhood (the surrounding $3\times3$ block of sites) as input to the network.
	
	Let $s^t \in \mathbb{R}^{\texttt{c} \times H \times W}$ denote the state of a system of height $H$ and width $W$ at time $t$. For each site $r$ (dropping boldface notation for vectors in this appendix), the NCA decides the the next state based on the the perception vector $z_r \in \mathbb{R}^{9\texttt{c}}$ which is acquired by concatenating all channel values in the $3 \times 3$ block centered at $r$. The NCA update at site $r$ is then calculated by applying two learned linear maps with a ReLU nonlinearity \footnote{The use of the ReLU nonlinearity does not bias the learned rules towards positive spins, since the weights of $W_1,W_2$ can have both positive and negative signs (and are initialized randomly from a distribution with zero mean). Indeed, out of the learned rules which function as robust memories, we observe no statistically significant bias between positive and negative spins.}:
	$$s^{t+1}_r = W_2 \sigma(W_1 z_r)$$
	where $\theta = [W_1, W_2]$ denote the trainable NCA parameters of the linear maps, and $\sigma(x) = \max(0,x)$ is the ReLU non-linear function.
	Unlike residual NCAs, where the network output is added to the current state, here the network output directly becomes the next state of each cell.
	For evaluating the model in binary setting, after each update, every positive value of the continuous state is mapped to $+\Delta$ and every negative value is mapped to $-\Delta$.
	
	\subsection{Training}
	We are interested in finding NCA rules that achieve robust memory under noisy dynamics. The noisy update rule is governed by three parameters: the noise probability $p$, the noise bias $\eta$, and the stochastic update rate $\alpha$. At each time step, every site is selected for update with probability $\alpha$. If selected, the site state is replaced either by the output of the NCA (with probability $1-p$) or by a random noise value (with probability $p$). The noise is drawn from $\{-\Delta, +\Delta\}$, where $\eta$ controls the bias towards $+\Delta$ or $-\Delta$. Specifically, $\eta=1$ means all noise is $+\Delta$, $\eta=-1$ means all noise is $-\Delta$, and $\eta=0$ produces unbiased noise. For NCAs with more than one channel, noise applied to site $r$ is defined to randomize the values of all channels at $r$. 
	
	In our experiments, we set $\alpha = 1$ to perform synchronous updates, which we use exclusively during training. For each batch, $\eta$ is sampled independently for every batch element from the uniform distribution $\mathcal{U}(-1, 1)$, while $p$ is sampled once per batch from $\mathcal{U}(0, p_{\mathrm{max}})$ with a fixed upper bound $p_{\mathrm{max}}$. The NCA is trained directly on continuous state values, with the binary mode used only for evaluation (as in the main text). We also find that setting $\mathtt{c}=1$---that is, dispensing with the hidden channels altogether---still allows robust memories to be readily found in training (although including hidden channels should allow for memories which persist to higher noise strengths). As such, we take $\mathtt{c}=1$ in our experiments. 
    
    The detailed training process is described in Algorithm~\ref{alg:nca-training}.

\begin{figure*}[t] 
	\begingroup
	\renewcommand{\figurename}{Algorithm} 
    \algorule
	\caption{	\label{alg:nca-training} Training procedure for Memory NCA}
    \algorule
	\endgroup
	
		\begin{algorithmic}[1]
			\State \textbf{Initialize Model:} $\mathrm{NCA}_\theta$ with \texttt{c} channels and parameters $\theta$.
			\State \textbf{Initialize optimizer: } Adam optimizer \texttt{optim} with learning rate \texttt{lr}.
			\State \textbf{Initialize pool:} $\{s_j, t_j\}_{j=1}^M$ with grid size $H \times W$. 
			Each $s_j$ is seeded with uniform spin states ($+\Delta$ or $-\Delta$) and large error patches (inverted spins).  
			Each $t_j$ is the corresponding target spin configuration ($\pm \Delta$).
			\For{epoch $= 1$ to $N_{\mathrm{epochs}}$}
			\State Randomly sample batch of size $B$ from the pool $\{s, t\}$.
			\State Reinitialize one batch element with a fresh seed-target pair.
			\State For each batch element, sample $\eta \sim \mathcal{U}(-1, 1)$.
			\State Sample $p \sim \mathcal{U}(0, p_{\max})$ (shared across batch).
			\State Sample roll-out length $K \sim \mathrm{Uniform}\{K_{\min}, K_{\max}\}$.
			\For{$k = 1$ to $K$}
			\State $s \gets \mathrm{Simulate}_{[p, \eta, \alpha]}(\mathrm{NCA}_\theta, s)$ \text{\texttt{\quad\quad     \scriptsize{noisy evolution of the state with NCA updates}}}
			\EndFor
			\State $L_{\mathrm{overflow}} = \ell_{1,2}\big(s, \mathrm{clip}(s, [-1, 1])\big)$ \text{\texttt{\quad\quad \scriptsize{overflow loss}}}
			\State $L_{\mathrm{memory}} = \ell_{1,2}\big(s, t\big)$ \text{\texttt{\;\;\quad\quad\quad\quad\quad\quad\quad \scriptsize{memory preservation loss}}}
			\State $L = L_{\mathrm{overflow}} + L_{\mathrm{memory}}$ \text{\texttt{\;\;\;\quad\quad\quad\quad\quad \scriptsize{total loss}}}
			\State $\theta \gets \texttt{optim}(\nabla_{\theta} L, \texttt{lr})$ 
			\text{\texttt{\;\quad\quad\quad\quad\quad\quad\quad \scriptsize{backpropagation followed by stochastic gradient descent}}}
			\State Replace failed cases in the batch (incorrect final sign) with new seeds.
			\State Write updated batch back to the pool.
			\EndFor
		\end{algorithmic}
        \algorule
\end{figure*}

	All of our experiments used to discover memory CA rules use the following settings:
	\begin{itemize}
		\item Spin encoding magnitude: $\Delta = 0.25$
		\item Number of channels: $\texttt{c} = 1$
		\item Fully-connected hidden neurons of NCA: 128
		\item Grid size: $H \times W = 64 \times 64$
		\item Pool size: 512 
		\item Update probability: $\alpha = 1.0$ (synchronous updates)
		\item Maximum noise probability: $p_{\mathrm{max}} = 0.05$
		\item Learning rate: $10^{-3}$ with decay at epochs $[2000, 4000, 6000]$ by a factor $0.3$
		\item Roll-out length: $K_{\min}, K_{\max} = 1, 8$
		\item Batch size: 8
		\item Training epochs: 8000
	\end{itemize}

	Once the NCA model is trained, we need to convert it into a binary CA that operates on the discrete state space $\{-1, 1\}$.
	To derive the deterministic binary cellular automaton corresponding to a trained NCA, we exhaustively simulate its response to all possible binary Moore neighborhoods with \texttt{c} channels.  
	There are $2^{9\texttt{c}}$ such distinct configurations.  
	Each configuration $b \in \{-1, +1\}^{9\texttt{c}}$ is first encoded by mapping $-1 \mapsto -\Delta$ and $+1 \mapsto +\Delta$, then updated using a single NCA step with $\alpha = 1$ and $p = 0$.  
	The sign of the center cell in the updated state is taken as the binary output for that configuration.  
	Collecting these outputs in lexicographic order produces a binary string of length $2^{9\texttt{c}}$, which fully specifies the induced deterministic binary rule. 
	In effect, training the NCA serves as a guided search through the vast space of possible binary CA rules, whose size is $2^{2^{9\texttt{c}}}$.

	In total, we trained $1000$ independent NCA models, each initialized with a different random seed, resulting in variations due to both random weight initialization and stochastic noise realizations during training.  
	Among these, only $37$ models produced binary CAs that functioned as robust memories after conversion.  
	This relatively low success rate suggests that our current training procedure is far from optimal.  
	We expect that substantial gains could be achieved by refining the loss function, exploring alternative training strategies, and tuning hyperparameters to increase the probability of converging to NCAs whose induced binary CAs exhibit robust memory behavior.

	\begin{table}[htbp]
		\centering
		\caption{Details of the CA rules studied in the main text, along with a few other representative examples. The convention for determining the rule string is given in \eqref{rulestring}. }
		\label{tab:rule_landscape}
		\resizebox{\textwidth}{!}{% <------ Don't forget this %
			\begin{tabular}{p{0.12\textwidth} p{.48\tw} p{0.25\textwidth}}
				\toprule
				rule name & {\rm rule string} & comments \\
				\midrule
				{\sf A} automaton & \tiny\seqsplit{00000000000000110000000000000011000000110011111100000011001111110000000000000011000000000000001100000011001111110000001100111111000000110011111100000011001111110011111111111111001111111111111100000011001111110000001100111111001111111111111100111111111111110000000000000011000000000000001100000011001111110000001100111111000000000000001100000000000000110000001100111111000000110011111100000011001111110000001100111111001111111111111100111111111111110000001100111111000000110011111100111111111111110011111111111111} & majority vote on five spins; minority domains eroded into pentagonal shapes; $\zt$ spin-flip and $x\lra y$ reflection symmetries \\ 
				\midrule 
				{\sf B} automaton & \tiny\seqsplit{00000000000000000000000000000001000000010000000100000001000101110000000100000001000000010001011100010111000101110001011101111111000000010000000100000001000101110001011100010111000101110111111100010111000101110001011101111111011111110111111101111111111111110000000000000001000000010001011100000001000101110001011101111111000000010001011100010111011111110001011101111111011111111111111100000001000000010001011100010111000101110001011101111111011111110001011100010111011111110111111101111111011111111111111111111111} &  minority domains eroded into moving triangle shapes \\ \midrule 
				{\sf C} automaton & \tiny\seqsplit{00000001000000010000000100000101000000010001010100000001000101010000000100000001000000010001011100000001000101110001011100111111000000010000000100000001000101110000000100000101000001010001111100000001000100110001011100011111000101110011111101111111111111110000000100000101000000010001011100000001000101010000010100011111000000010001011100010111001111110001111101111111111111111111111100000001000101110000111111111111000101010111111101111111111111110001111101111111111111111111111111111111111111111111111111111111} &  minority domains eroded by squeezing along oblique axis \\ 
				\midrule 
				{\sf D} automaton & \tiny\seqsplit{00000000000000000000000000000001000000000000000000000001000101110000000000000001000000010001011100000001000001010001011101111111000000000000000000000001000101110000000000000001000101110111111100000001000001110001011101111111000001110001111101111111111111110000000100000101000000010001011100000101000101110000010100010111000101110101111100010111011111110101011101111111010101110111111100000101000101110001010101010111000101110111111100010111011111110101111101111111010101110111111101111111111111110111111111111111} & exponentially many steady states at $p=0$; noise-stabilized order. \\ \midrule 
				destabilized Toom & \tiny\seqsplit{10000000000000000011001100110011001100110011001111111111111111110000000000000000001100110011001100110011001100111111111111111111000000000000000000110011001100110011001100110011111111111111111100000000000000000011001100110011001100110011001111111111111111110000000000000000001100110011001100110011001100111111111111111111000000000000000000110011001100110011001100110011111111111111111100000000000000000011001100110011001100110011001111111111111111110000000000000000001100110011001100110011001100111111111111111110} & no uniform absorbing states; robust for both synchronous and asynchronous updates \\ 
				\midrule  
				shifted Toom & \tiny\seqsplit{00000000000000000000000000000000000000000000000000000000000000000000000011111111000000001111111100000000111111110000000011111111000000001111111100000000111111110000000011111111000000001111111111111111111111111111111111111111111111111111111111111111111111110000000000000000000000000000000000000000000000000000000000000000000000001111111100000000111111110000000011111111000000001111111100000000111111110000000011111111000000001111111100000000111111111111111111111111111111111111111111111111111111111111111111111111} & Toom's rule with the voting region translated by $-\uvx-\uvy$. Equivalent to Toom's rule for sync updates; under async updates does not linearly erode domains with biased noise \\ 
				\midrule 
				grower & \tiny\seqsplit{01101001011000010000000100000101000000010001010100000001000101010000000100000001000000010001011100000001000101110001011100111111000000010000000100000001000101110000000100000101000001010001111100000001000100110001011100011111000101110011111101111111111111110000000100000101000000010001011100000001000101010000010100011111000000010001011100010111001111110001111101111111111111111111111100000001000101110000111111111111000101010111111101111111111111110001111101111111111111111111111111111111111111111111111111111111} & two uniform absorbing states; 0 domains send out small propagating domains of 1s \\ 
				\bottomrule
			\end{tabular}
		}
		
	\end{table} 
	
	\section{Thresholds under asymmetric majority voting}\label{app:majproof}
	
	In this appendix, we prove Theorem~\ref{thm:majtheorem} from the main text. Before we begin, a remark on notation. In this appendix, we will work with spins valued in $\{0,1\}$, and will think of spins as Boolean variables, with 0 associated to \texttt{false} and 1 to \texttt{true}. As in the main text, we will work with a general $2$-dimensional CA rule $\mca$ which acts on a spin at site $\bfr$ in a way determined by the spins on a (in this appendix, potentially arbitrarily-shaped) neighborhood $\mcr_\bfr$ of $\bfr$ (we will just write $\mcr$ when we just want to indicate a general neighborhood without reference to $\bfr$). $\S_\mcr \equiv \{0,1\}^{|\mcr|}$ will be used to denote the set of spin configurations on $\mcr$.  
	
	Describing the proof of Theorem~\ref{thm:majtheorem} will require a few definitions. The first is a property of Boolean functions that is useful for proving the stability of certain kinds of cellular automata: 
	
	\begin{definition}[monotonicity]
		A CA rule $\mca \, :\, \S_\mcr \ra \{0,1\}$ is {\it monotonic} if increasing the number of 1s in the input cannot decrease the number of 1s in the output: 
		\be x \prec y \implies \mca(x) \leq \mca(y) \quad \forall \,\, x,y \in \S_\mcr, \ee 
		where $x\prec y $ means $x_i \leq y_i \, \, \forall \,\, i \in \{1,\dots,\abs{\mcr}\}$. 
	\end{definition}
	
	We will also need the notion of the dual of a  CA rule: 
	\begin{definition}[dual CA]
		For a CA rule $\mca \, : \, \S_\mcr \ra \{0,1\}$, the {\it dual rule} $\mca^\neg$ is defined by negating $\mca$'s rule table: 
		\be \mca^\neg(x) = \neg \mca(\neg x) \quad \forall \, \, x \in \S_\mcr.\ee 
		We will say $\mca$ is self-dual if $\mca^\neg = \mca$. 
	\end{definition}
	One may verify that $\mca^\neg$ is monotonic if $\mca$ is. 
	We next need the definition of an eroder: 
	\begin{definition}[eroder]
		Consider any input state $s$ on an infinite square grid which differs from the all-0s state only in a finite region $\mcb$: 
		\be s_\bfr = 0 \, \, \forall \, \, \bfr \in  \zz^2 \setminus \mcb.\ee 
		Then $\mca$ is an eroder if there exists a finite constant $a$ (in general dependent on $|\mcb|<\infty$) such that $\mca^t(s)$ is the all-0s state for all $t \geq a$. 
	\end{definition}
	For such a state $s$, the {\it damage} ${\sf dam}(s)$ of $s$ will refer to the set where $s$ is nonzero: 
	\be \dam(s) = \{ \bfr \, : \, s_\bfr = 1 \}.\ee 
	An eroder is thus a CA which is guaranteed to erase a finite amount of damage in a finite time, so that there exists a constant $t$ such that $\dam(\mca^t(s)) = \emptyset$. 
	
	Finally, we need a definition of a robust memory. Ours will be slightly idiosyncratic, but will suffice for the present purposes:  
	\begin{definition}[robust memories]
		For a CA $\mca$ and a real number $p \in [0,1]$, the noisy CA $\mca_p$ is defined to be a CA which, at each step, performs the update defined by $\mca$ with probability $1-p$, and does something else with probability $p$ \footnote{To avoid getting bogged down in details, we will not give a very explicit description of the class of allowed noise models; the ones relevant for the present work are known as ``$p$-bounded error models.'' In the present context, it suffices to let $\mca_p$ implement $\mca$ with probability $1-p$, and to flip the spin in question to $0$ or $1$ with probability $(1\pm \eta)/2$. }.
		$\mca$ is said to be a {\it robust memory} if there exists a constant $p_c$ such that for all $p\leq p_c$, there exist initial states $s_0,s_1$ such that, after being evolved under $\mca_p$, the magnetization of $s_0$ ($s_1$) remains negative (positive) for a time that diverges in the thermodynamic limit. 
	\end{definition}
	
	The proof of Theorem~\ref{thm:majtheorem} is a consequence of the following powerful result of Toom, which we will state in a slightly informal way:  
	\begin{theorem}[Toom's monotone eroder theorem \cite{toom1980stable}]
		If $\mca$ is monotonic and both $\mca$ and $\mca^\neg$ are eroders, then $\mca$ is a robust memory when updated synchronously. 
	\end{theorem}
	While being a monotone eroder is a sufficient condition for $\mca$ to be a robust memory under synchronous updates, the CA rules found by our NCA framework demonstrate that it is by no means necessary. It is also important to note that this theorem does not say anything about robustness under asynchronous updates. The authors are currently unaware of whether or not there exists a simple sufficient condition for an asynchronous automaton to be a robust memory.
	
	Suppose $\mca$ performs a majority vote over a region $\mcv\subset\mcr$. The majority vote of $n$ spins is clearly a monotonic function, and its symmetry guarantees that $\mca$ is self-dual. Therefore, to prove Theorem~\ref{thm:majtheorem}, it suffices to prove results about when majority voting automata are eroders. This is done in the following theorem, the result of which implies Theorem~\ref{thm:majtheorem}:
	
	\begin{theorem}[asymmetric majority voting and erosion]\label{thm:majtheorem_app}
		Let $\sfN_a(\bfr) = \{ \bfr' \, : \,  |\bfr'-\bfr|_\infty \leq a\}$ \footnote{$\sfN_1(\bfr)$ is known as the {\it Moore neighborhood} of $\bfr$.}.
		Consider a set of sites $\mcv$ and a CA rule which sends $s_\bfr$ to the majority vote of the spins on the set $\bfr + \mcv$: 
		\be \mca \, : \, s_\bfr \mt  \maj \( \{ s_{\bfr+\bfdel} \, : \, \bfdel \in \mcv\}\),\ee 
		where $\bfr + \mcv = \{\bfr+\bfdel\, : \, \bfdel \in \mcv\} \subset {\sf N}_a(\bfr + \bfd)$ for some constant offset vector $\bfd$, and $|\mcv| \in 2\zz+1$.
		
		If $\mcv$ is invariant under a $\pi$ rotation symmetry $C_\pi$, then $\mca$ is not an eroder. On the other hand, if $\mcv$ is not invariant under any $\pi$ rotation symmetry, then $\mca$ is an eroder if $a=1$.
	\end{theorem} 
	
	Note that we have phrased the result as an iff in the main text (c.f. Theorem~\ref{thm:majtheorem}), since there we assumed $\bfd = \bfzero,\, \mcv \subset \sfN_1(\bfzero)$. 
	Note also that by composing $\mca$ with a translation $T_\bfd$ along $\bfd$, we can shift $\mcv$ so that $\{\bfr + \bfdel \, : \, \bfdel \in \mcv\} \subset \sfN_a(\bfr)$. The composition with $T_\bfd$ is innocuous for synchronous updates \footnote{This shift is innocuous for synchronous updates only. Composing with a translation can drastically modify the behavior of a CA if the updates are performed asynchronously.}, since if $\mca$ is an eroder then so too is $\mca \circ T_\bfd$. In what follows we will assume that this shift has been performed. 
	
	We will first show that $\mca$ is not a robust memory if $\mcv$ is invariant under a $\pi$ rotation (regardless of $a$). To avoid getting bogged down in details, our discussion will be rather laconic. 
	
	\begin{proof} 	
		The basic point is that $C_\pi$ symmetry implies that domain walls of uniform slope cannot move under the dynamics. This follows simply from the fact that by translation symmetry, an infinitely long domain wall of uniform slope must, under $\mca$, either remain at rest or move in a definite direction while maintaining its shape. The latter possibility is forbidden by $C_\pi$ symmetry: if a domain wall of sign $s$ moves along $\uva$, then a domain wall of sign $-s$ must move along $-\uva$. Since $\mca = \mca^\neg$, this is only possible if $\uva = - \uva$, and hence the domain wall must in fact remain stationary.  
		
		This argument shows that the error-correcting dynamics of $C_\pi$-symmetric majority voting must be driven by domain wall curvature. Since the dynamics is deterministic, synchronous, and occurs on a lattice, a minority domain whose boundary has everywhere sufficiently small curvature will never be erased under $\mca$. This can be shown formally by an exhaustive analysis of the way the dynamics treats the intersection of a domain wall with slope $s$ with a domain wall of slope $s'$. As long as $s,s'$ do not differ by too much, the junction between the two domain walls can be shown to also be invariant under $\mca$. 
		
		A faster but more abstract proof uses Toom's characterization of eroders in terms of minimal zero sets, which is explained nicely in Ref.~\cite{ponselet2013phase}. While the details will not be provided here, the basic argument is to first note that if $\mcv$ is invariant under a $\pi$ rotation about some lattice point and contains an odd number of points then it has to have a point $\bfr_0$ at its center, and to then show that the symmetry implies that the convex hulls of each minimal zero-set must intersect at $\bfr_0$. 
	\end{proof}

	We now turn to showing that if $\mcv$ is not invariant under a $\pi$ rotation, and if $a=1$, then $\mca$ is an eroder \footnote{A proof in terms of minimal zero sets is again possible, but we will adopt a more elementary approach. }. 
	To do this, we need to first introduce some notation. Define the lines 
	\be \g_x^\pm(\bfr) = \{ \bfr + (2t, \pm t) \},\qq \g_y^\pm(\bfr)  = \{ \bfr + (t,\pm 2t) \},\qq t \in \rr.  \ee 
	The explicit dependence on the basepoint $\bfr$ will be omitted from our notation when its presence is not essential. 
	
	\begin{proposition}[motion of $\g_{x/y}^\pm$ domain walls]\label{prop:dw_motion} 
		Consider domain walls aligned along the lines $\g_x^\pm(\bfr), \g_y^\pm(\bfr)$ for some $\bfr$, and suppose the CA rule in question has no $C_\pi$ symmetries, and is defined with $\mcv \subset \sfN_1$, viz., with $a =1$ (c.f. the statement of Theorem~\ref{thm:majtheorem_app}). Then under the dynamics defined by $\mca$: 
		\begin{enumerate}
			\item domain walls aligned along $\g_x^\pm$ move along either $+\uvx$ or $-\uvx$ with speed $v=1$, or remain motionless. 
			\item those aligned along $\g_y^\pm$ move along either $+\uvy$ or $-\uvy$ with speed $v=1$, or remain motionless.
		\end{enumerate}
	\end{proposition}
	\begin{proof}
		This follows by translation invariance and the fact that the box size of $a=1$ is small enough to prevent multiple ``steps'' (c.f. Fig.~\ref{fig:lboxes}) in the lattice-scale domain wall from being seen by one update region. The only possibility is for the steps in a $\g_{a}^\pm$ domain wall to move along $+\uva$, $-\uva$, or remain motionless \footnote{That $\g_a^\pm$ domain walls cannot move normal to $\uva$ can be seen by contradiction. Suppose such domain walls could move. One may verify by inspection that this would only be possible if $\mca$ performed a majority vote along a single row or column of the Moore neighborhood $\sfN_1$. Such a situation is ruled out by our assumption of broken $C_\pi$ symmetry. }.
	\end{proof}

	To prove the remaining part of Theorem~\ref{thm:majtheorem_app}, we will use the following Lemma: 
	
	\begin{figure*}
		\includegraphics[width=.75\tw]{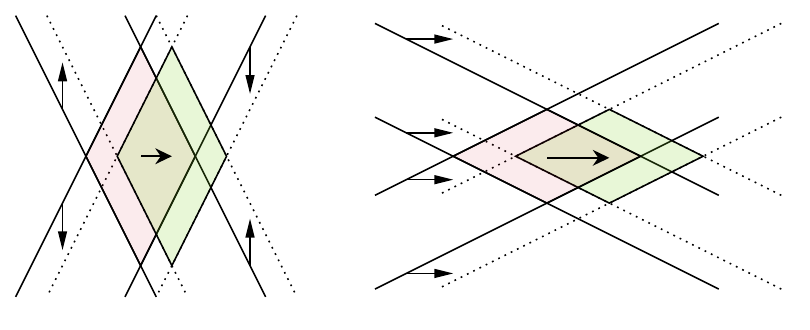} \caption{\label{fig:moving_erosion} Some considerations relevant to the proof of Lemma~\ref{lemma:moving_erosion}. In the left panel, domain walls along $\g_y^\pm$ move along $\mp \uvy$, with the result that the bounded region shaded in red ($\mcd'(t)$ in the proof) moves along $+\uvx$ to the location of the green region. The right panel is analagous, and shows the motion of $\mcd(t)$ in the scenario where domain walls along $\g_x^\pm$ move along $+\uvx$.   }
	\end{figure*}
	
	\begin{lemma}[moving domain walls implies erosion]\label{lemma:moving_erosion} 
		Let $\mca$ be a monotonic, self-dual CA rule. Suppose that domains oriented along $\g_{x/y}^\pm$ move along $\uvx/\uvy$ with speed $v_{x/y}^\pm \in \{0,1,-1\}$ under $\mca$, with at least on of the $v_{x/y}^\pm \neq 0$. Then $\mca$ is an eroder. 
	\end{lemma} 
	\begin{proof} 
		This follows from monotonicity and a simple geometric argument, part of which is illustrated in Fig.~\ref{fig:moving_erosion}. 
		
		Consider an initial state $s$ with damage $\dam(s)$ contained in a finite region. Increase the number of 1s in the initial state to obtain the state $s_\mcd$ with damage $\dam(s_\mcd) = \mcd$, where $\mcd \supset \dam(s)$ is a finite-sized diamond-shaped region bounded by lines $\g_x^+(\bfr^+_{1,2}), \g_x^-(\bfr^-_{1,2})$ with appropriate basepoints $\bfr^\pm_{1,2}$ (e.g. the shaded red region in the right panel of Fig.~\ref{fig:moving_erosion}). Under evolution by $\mca$, monotonicity then guarantees that 
		\be \dam(\mca^t(s)) \subset \mcd(t),\ee 
		where $\mcd(t)$ is the diamond-shaped region bounded by the lines $\g_x^\pm(\bfr^\pm + v_x^\pm t)$. One may verify the following statements by inspection of Fig.~\ref{fig:moving_erosion}: 
		\begin{enumerate}
			\item If $v_x^+, v_x^- = \pm1$, then $\mcd(t)$ moves along $+\uvx$ with speed 1;
			\item If $v_x^+ = - v_x^- = 1$, then $\mcd(t)$ moves along $-\uvy$ with speed $1/2$;
			\item If $v_x^+ = 1, v_x^- =0$, $\mcd(t)$ moves parallel to $\g_x^-$,  
		\end{enumerate}
		together with analogous symmetry-related statements ($\mcd(t)$ of course does not move if $v_x^+ = v_x^- = 0$). 
		
		The desired result then follows by applying the same logic to the lines $\g_y^\pm$. We increase the number of 1s in $s$ to form that state $s_{\mcd'}$ with damage $\dam(s_{\mcd'}) = \mcd'$, where $\mcd' \supset \dam(s)$ is a finite-sized diamond-shaped region bounded by lines $\g_y^+(\bfr^+_{1,2}), \g_y^-(\bfr^-_{1,2})$ with appropriate basepoints $\bfr^\pm_{1,2}$ (e.g. the shaded red region in the left panel of Fig.~\ref{fig:moving_erosion}). Monotonicity then implies $\dam(\mca^t(s)) \subset \mcd'(t)$. Similarly to the above, $\mcd'(t)$ moves in different ways depending on the values $v_y^\pm$. In particular, one may verify that there exists a time  $t_\star$ of order $\sqrt{\max(|\mcd|,|\mcd'|)}$ such that  
		\be \mcd(t) \cap \mcd'(t) = \emptyset\quad \forall \, \, t > t_\star,\ee 
		unless $v_x^\pm = v_y^\pm = 0$ (which contradicts the assumption of the lemma). This follows since even if $\mcd(t),\mcd'(t)$ move in the same direction, they do so with different speeds. 
		Since $\dam(\mca^t(s)) \subset \mcd(t) \cap \mcd'(t)$, this shows the claim. 	
	\end{proof}

	We can now finish the proof of Theorem~\ref{thm:majtheorem_app}. By the result of Lemma~\ref{lemma:moving_erosion}, we only need show that at least one of the domain walls parallel to the lines $\g_{x/y}^\pm$ moves under $\mca$ if $\mca$ lacks a $C_\pi$ symmetry. 
	
	\begin{figure*}
		\includegraphics[width=.4\tw]{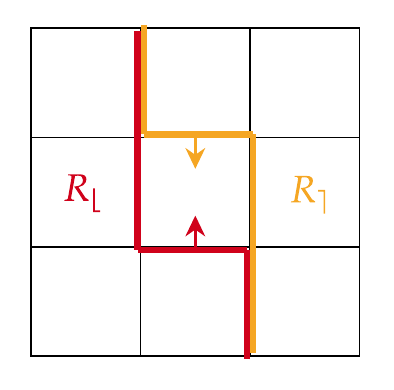}
		\caption{\label{fig:lboxes} An example of a partitioning of $\sfN_1(\bfzero)$ relevant for the proof of theorem~\ref{thm:majtheorem_app}. The red and orange lines indicate boundaries of domain walls parallel to $\g_y^-$; $R_\lfloor$ is the region to the left of the red line and $R_\rceil$ is the region to the right of the orange line. If $\bfzero\not\in\mcv$, then either the step edge of the domain wall marked in red must move up as indicated by the arrow, or the edge of the domain wall marked in orange must move down.   }
	\end{figure*}
	
	\begin{proof} 
		Consider first the case where $\bfzero \not\in \mcv$ (which in the present setting then automatically implies that $\mcv$ possesses no $C_\pi$ rotation symmetry). In this case, {\it every} domain wall aligned along one of the lines $\g_{x/y}^\pm$ must move under $\mca$. This is simply because $\sfN_1(\bfzero) \setminus \{\bfzero\}$ can be partitioned in various ways into two sets each containing four sites, which by $|\mcv|\in 2\zz+1$ imply that the ``step edges'' on each such domain wall must move under $\mca$. An example of this is illustrated in Fig.~\ref{fig:lboxes}, where $\mcv$ is partitioned into a left part $R_\lfloor$ and a right part $R_\rceil$. Since 
		\be |\mcv\cap R_\lfloor| \neq |\mcv\cap R_\rceil|,\ee   the step edge of the domain wall must move either up or down. Identical considerations apply for domain walls parallel to other $\g_{x/y}^\pm$ lines. 
		
		Now consider the case where $\bfzero \in \mcv$. It is then possible for domain walls parallel to the lines $\g_{x/y}^\pm$ to remain fixed under $\mca$ (since it is, e.g., now possible to have $|\mcv\cap R_\lfloor| = |\mcv\cap R_\rceil|$). However, we claim that if {\it all} such domain walls are fixed, then $\mcv$ is invariant under a $C_\pi$ rotation about $\bfzero$. This can be shown using the same partitioning of $\sfN_1(\bfzero) \setminus \{\bfzero\}$. Consider the partitioning shown in Fig.~\ref{fig:lboxes}, together with an analogous partitioning obtained by reflection through the $y$ axis, with analogous regions $R_{\lceil}, R_{\rfloor}$. If domain walls along both $\g_y^+$ and $\g_y^-$ are to be fixed under $\mca$, then we must have both $|\mcv\cap R_\lceil| = |\mcv\cap R_\rfloor|$ and $|\mcv\cap R_\lfloor| = |\mcv\cap R_\rceil|$. Subtracting these equations from one another implies 
		\be |\mcv \cap \{\uvy\}| = |\mcv \cap \{-\uvy\}|.\ee 
		Identical arguments for other types of rotation- and reflection-related partitions show that if all of the $\g_{x/y}^\pm$ domain walls are fixed under $\mca$, then 
		\be |\mcv \cap \{ \bfa\}| = |\mcv \cap \{-\bfa\}| \,\, \forall \, \, \bfa \in \sfN_1(\bfzero) \setminus \{\bfzero\}.\ee 
		This is however only possible if $\mcv = C_\pi(\mcv)$, a contradiction. 
	\end{proof}
	
	Finally, we note that the restriction to $a=1$ in the theorem statement is in fact necessary. To show this, we explicitly construct a rule for which $\mcv$ has no $C_\pi$ symmetry but which fails to be an eroder. An example with $a=2$ is 
	\be \mcv = \{\bfzero, \uvx-\uvy,-\uvx-\uvy,2\uvx+2\uvy,-2\uvx+2\uvy\}.\ee 
	One may directly verify that this rule fails to erode square minority domains of large enough linear size $\ell$ (taking $\ell > 10$ suffices).

	\section{Mean field theory}\label{app:meanfield}
	
	In this Appendix, we perform a mean-field analysis of the CA dynamics studied in the main text. We will return to letting the spins $s_\bfr$ be valued in $\{-1,1\}$. 
	
	We begin by writing down a master equation for the time evolution of the expected magnetization $\lan s_\bfr(t)\ran$ at site $\bfr$, assuming as usual that CA updates with rule $\mca$ are applied with probability $1-p$, and noise with bias $\eta$ is applied with probability $p$. We work in the setting of asynchronous updates, where each cell is updated according to independent Poisson processes of rate $1/dt$. In this setting, $\lan s_\bfr(t)\ran$ evolves as 
	\be\label{master} \lan s_\bfr(t+dt)\ran = (1-dt)\lan s_\bfr(t)\ran + dt \((1-p)\lan \mca(s|_{\mcr_\bfr}(t))\ran + p\eta\). \ee 
	This equation is exact, but contains $\lan \mca(s|_{\mcr_\bfr}(t))\ran$, which in general is a complicated sum of $|\mcr|$-point functions of the $s_\bfr(t)$. In mean-field, we simplify this by assuming that the probability distribution over the $s_\bfr(t)$ is entirely determined by its 1-point marginals, viz., we assume that for all $\bfr\neq\bfr'$, 
	\be \lan s_\bfr(t) s_{\bfr'}(t) \ran = \lan s_\bfr(t) \ran\lan s_{\bfr'}(t) \ran .\ee 
	We will use the symbol $m(t)$ to denote $\lan s_\bfr(t)\ran$, which we assume to be independent of $\bfr$ (which holds for the steady states of all of the CA of interest to us).  
	With this notation, we may write (keeping the time dependence implicit)
	\bea \label{mfa} \lan \mca(s|_{\mcr_\bfr}) \ran & = \sum_{w\in \S_\mcr} \mca(w)\prod_{\bfr \in \mcr} \left\lan \frac{1+w_\bfr s_\bfr}2 \right\ran \\ & = \sum_{w \in \S_\mcr} \mca(w) \prod_{\bfr \in \mcr} \frac{1+w_\bfr m }2 \\ 
	& \equiv \mca_{MF}(m).\eea 
	\eqref{master} then becomes 
	\be \label{mf}\p_t m =p(\eta - m) + (1-p)\( \mca_{MF}(m)- m \). \ee
	Since the RHS is a function of the single real variable $m$, we may always write 
	\be \p_t m = - \frac{\partial F_{\rm eff}}{\partial m} \ee 
	for an appropriate free energy $F_{\rm eff}$. Explicitly, 
	\be \label{mfexp} F_{\rm eff} = -p\eta m + \frac{m^2}2 -(1-p)  \int dm \, \mca_{MF}(m) = \sum_{n=1}^{|\mcr|+1} \frac{f_n}{n!} m^n.\ee 
	The coefficients $f_n$ can be computed by inspecting powers of $m$ after expanding the expression for $\mca_{MF}(m)$ in \eqref{mfa}. For $f_1$, we have 
	\be f_1 = -p \eta - \frac{1-p}{2^9} \sum_{w \in \S_\mcr} \mca(w). \ee 
	% For $f_2$, 
	% \bea f_2 & = 1 - (1-p)  \sum_{w\in \S_\mcr} \sum_{\bfr\in \mcr} w_\bfr \mca(w) \\ 
	% & = 1 - (1-p) |\mcr| (2\sfP[\mca(w) = w_\bfr] - 1),\eea  
	For $f_2$, 
	\bea f_2 & = 1 - \frac{1-p}{2^9}\sum_{w\in \S_\mcr} \sum_{\bfr\in \mcr} w_\bfr \mca(w) \\ 
	& = 1 - (1-p) |\mcr| (2\sfP[\mca(w) = w_\bfr] - 1),\eea  
	where $\sfP[\mca(w) = w_\bfr]$ is the probability that a randomly chosen input cell in $\mcr$ agrees with the output $\mca(w)$, with $w$ drawn uniformly on $\S_\mcr$.
	
	Out of all the $f_n$, only the sign of $f_2$ can change when $p$ is varied. The point where $f_2=0$ defines the mean-field prediction of the critical point in the case where a $\zt$ symmetry is present; in this case the mean-field prediction for the critical noise strength $p_c$ is 
	\be p_c = 1 - \frac1{|\mcr|(2\sfP_{\rm match} - 1)},\ee 
	where we have used the shorthand $\sfP_{\rm match} = \sfP[\mca(w) = w]$. Taking $p_c \ra 0$, we see that mean-field theory predicts an ordered phase only if 
	\be \sfP_{\rm match} > \frac12 \( \frac1{|\mcr|} + 1 \),\ee 
	which with high probability will not be satisfied for a random rule.
	For CA defined on the Moore neighborhood, order is possible in mean-field only if $\sfP_{\rm match} > 5/9$. 
	
	Explicit expressions are easy to arrive at in the case where $\mca$ performs a majority vote; in this case, mean-field always predicts a transition at nonzero $p$. Explicit polynomial expressions for the majority votes of $n$ spins are easy enough to work out. For example,  
	\bea \maj(m_1,m_2,m_3) & = \frac12\(S_1(\{m_i\}) - S_3(\{m_i\})\) \\ 
	\maj(m_1,m_2,m_3,m_4,m_5) & = \frac18\(3S_1(\{m_i\}) - S_3(\{m_i\}) + 3S_5(\{m_i\})\), \eea 
	where $S_k(\{m_i\}) = \sum_{i_1 < \cdots < i_k} m_{i_1} \cdots m_{i_k}$ is the symmetric sum of degree $k$ in the $m_i$. Taking $m$ to be uniform then gives an explicit expression for the coefficients $f_n$ appearing in \eqref{mfexp}. Letting $F_n$ be the mean-field free energy of an automaton which performs a majority vote on an odd number $n$ of spins, we find 
	\bea F_3 & = -p \eta m + \frac14(3p-1)m^2 + \frac{1-p}8 m^4 \\ 
	F_5 & = -p\eta m + \frac1{16}(15p -7)m^2 + \frac{5(1-p)}{16} m^4 \\ 
	F_7 & = -p\eta m + \frac1{32}(35p - 19)m^2 + \frac{35(1-p)}{64} m^4 .\eea 
	Giving mean-field estimates of $p_{c,MF} = 1/3,7/15,19/35$, for $n=3,5,7$, respectively (with $p_{c,MF} \ra 1$ as $n\ra \infty$). 
	
	Interestingly, almost none of the CA rules found by our NCA framework are stable in mean-field. In particular, the free energies of the non-majority-vote rules studied in the main text are 
	\bea F_{\sf B} & = -\eta pm + \frac{107 - 75p}{64} m^2 - \frac{45(1-p)}{64}m^4 + \frac{3(1-p)}8m^6 - \frac{15(1-p)}{128}m^8 + \frac{1-p}{64}m^{10}\\ 
	F_{{\sf C}} & = -\(\eta p + \frac{1-p}{128}\)m + \frac{423 - 295p}{256}m^2 + \frac{1-p}{64}m^3 - \frac{171(1-p)}{256}m^4 + \cdots + \frac{17(1-p)}{1280}m^{10} \\ 
	F_{{\sf D}} & = - \( \eta p - \frac{19(1-p)}{256}\) m + \frac{843- 587p}{512}m^2 - \frac{19(1-p)}{192}m^3 - \frac{167(1-p)}{256}m^4 + \cdots + \frac{7(1-p)}{640} m^{10}.
	\eea 
	$F_{\sf B}$ and $F_{\sf C}$ have a unique minimum for all $p,\eta$, and $F_{\sf B}$ is symmetric under $m \mt -m$ despite this not being an exact symmetry of $\sfB$'s rule table.  $F_{\sf D}$ on the other hand has multiple minima, one of which is at $m>1$. For $\eta=0$, multiple minima onset at $p_{c,MF} \approx 0.88$, a gross overestimate of $p_c$. 
	
\end{document}